\documentclass[conference,10pt]{IEEEtran}
\IEEEoverridecommandlockouts
\usepackage{cite}
\usepackage{amsmath,amssymb,amsfonts}
\usepackage{graphicx}
\usepackage{textcomp}
\usepackage{xcolor}
\def\BibTeX{{\rm B\kern-.05em{\sc i\kern-.025em b}\kern-.08em
    T\kern-.1667em\lower.7ex\hbox{E}\kern-.125emX}}

\usepackage{soul}
\usepackage{url}
\usepackage[hidelinks]{hyperref}
\usepackage[utf8]{inputenc}
\usepackage[small]{caption}
\usepackage{graphicx}
\usepackage{booktabs}
\usepackage{amsthm}
\usepackage{pgfplots}
\usepackage{tikz}
\usepackage{algorithm}
\usepackage[noend]{algpseudocode}
\usepackage{multirow}
\usepackage{subcaption}
\usepackage{empheq}

\def\BState{\State\hskip-\ALG@thistlm}
\makeatother

\def\x{{\underline{x}}}
\def\G{{\mathcal{G}}}
\def\A{{\mathbf{A}}}
\def\L{{\mathbf{L}}}
\def\S{{\mathbf{S}}}
\def\C{{\mathbf{C}}}
\def\B{{\mathbf{B}}}
\def\P{{\mathbf{P}}}
\def\U{{\mathbf{U}}}
\def\X{{\mathbf{X}}}
\def\Y{{\mathbf{Y}}}
\def\D{{\mathbf{D}}}

\DeclareMathOperator*{\argmin}{arg\,min} 

\newtheorem{theorem}{Theorem}
\theoremstyle{definition}
\newtheorem{definition}{Definition}
\theoremstyle{definition}
\newtheorem{problem}{Problem}

\begin{document}

\title{Piecewise Stationary Modeling of Random Processes Over Graphs With an Application to Traffic Prediction}

\author{\IEEEauthorblockN{Arman Hasanzadeh, Xi Liu, Nick Duffield and Krishna R. Narayanan} \IEEEauthorblockA{Department of Electrical and Computer Engineering \\
Texas A\&M University \\
College Station, Texas 77840\\
Email: \textit{\{armanihm, duffieldng, krn\}@tamu.edu, xiliu.tamu@gmail.com}}}

\maketitle

\begin{abstract}
Stationarity is a key assumption in many statistical models for random processes. With recent developments in the field of graph signal processing, the conventional notion of wide-sense stationarity has been extended to random processes defined on the vertices of graphs. It has been shown that well-known spectral graph kernel methods assume that the underlying random process over a graph is stationary. While many approaches have been proposed, both in machine learning and signal processing literature, to model stationary random processes over graphs, they are too restrictive to characterize real-world datasets as most of them are non-stationary processes. In this paper, to well-characterize a non-stationary process over graph, we propose a novel model and a computationally efficient algorithm that partitions a large graph into disjoint clusters such that the process is stationary on each of the clusters but independent across clusters. We evaluate our model for traffic prediction on a large-scale dataset of fine-grained highway travel times in the Dallas--Fort Worth area. The accuracy of our method is very close to the state-of-the-art graph based deep learning methods while the computational complexity of our model is substantially smaller.
\end{abstract}


\section{Introduction}
Stationarity is a well-known hypothesis in statistics and signal processing which assumes that the statistical characteristics of random process do not change with time \cite{papoulis1965probability}. Stationarity is an important underlying assumption for many of the common time series analysis methods \cite{dicky1979unitroot, Brockwell2016ts, Cryer2008ts}. With recent developments in the field of graph signal processing \cite{shuman2013emerging, moura2014bigdata,ortega2018graph}, the concept of stationarity has been extended to random processes defined over vertices of graphs \cite{loukas2016stationary, marques2017stationary, perraudin2017tstationary, perraudin2017sspog}.
A random process over a graph is said to be graph wide-sense stationary (GWSS) if the covariance matrix of the process and the shift operator of the graph, which is a matrix representation of the graph (see Section \ref{sec: stationary}), have the same set of eigenvectors.

Incidentally, GWSS is the underlying assumption of spectral graph kernel methods which have been widely used in machine learning literature to model random processes over graphs \cite{smola2003kernels,sollich2009kernelgp,zhu2005nonparametric,urry2013random,avrachenkov2018similarities}. The core element in spectral kernel methods is the kernel matrix that measures similarity between random variables defined over vertices. This kernel matrix has the same set of eigenvectors as the Laplacian matrix (or adjacency matrix) while its eigenvalues are chosen to be a function of eigenvalues of the Laplacian matrix (or adjacency matrix). Therefore, the kernel matrix and the shift operator of the graph share eigenvectors which is the exact definition of GWSS.

While the stationarity assumption has certain theoretical and computational advantages, it is too restrictive for modeling real-world big datasets which are mostly non-stationary. In this paper, we propose a model that can deal with non-stationary covariance structure in random processes defined over graphs. The method we deploy is a novel and computationally efficient graph clustering algorithm in which a large graph is partitioned into smaller disjoint clusters. The process defined over each cluster is stationary and assumed to be independent from other clusters. To the best of our knowledge, our proposed clustering algorithm, stationary connected subgraph clustering (SCSC), is the first method to address this problem. 
Our model renders it possible to use highly-effective prediction techniques based on stationary graph signal processing for each cluster \cite{loukas2017evolution}. An overview of our proposed clustering algorithm is shown in Fig. \ref{fig: method}.

\begin{figure*}[t]\label{fig: method}
\centering
\includegraphics[width=\textwidth, keepaspectratio]{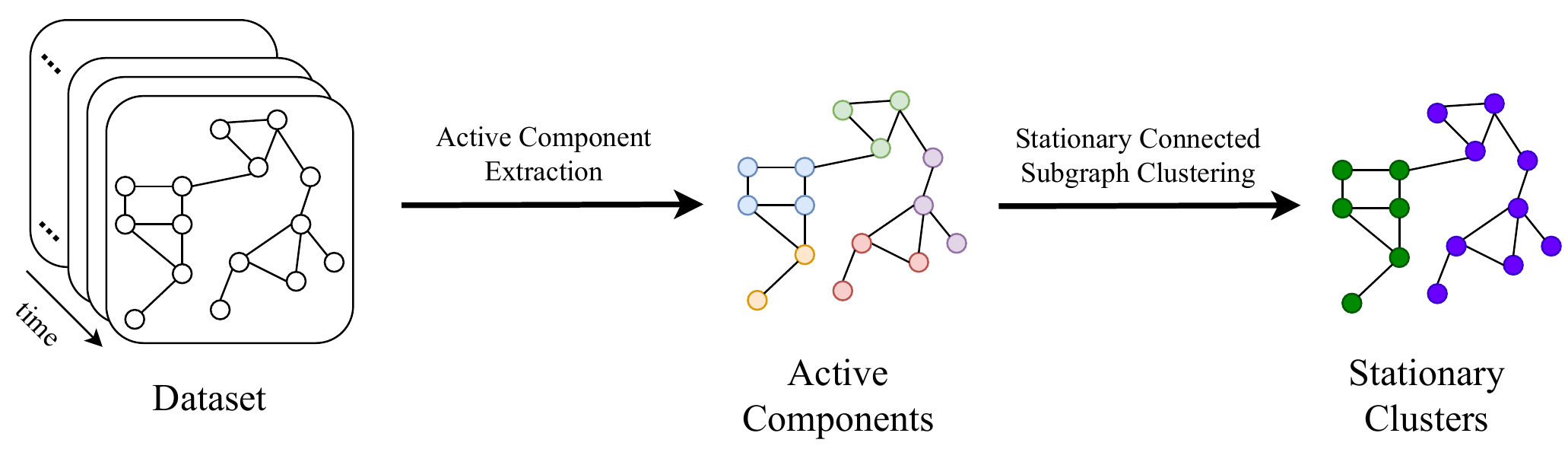}
\caption{An overview of our proposed stationary clustering algorithm. Given historical observation, i.e. time-series defined over nodes of a graph, first we extract active components. Next, our proposed stationary connected subgraph clustering (SCSC) tries to merge adjacent active components to form stationary subgraphs.}
\label{fig: method}
\end{figure*}

Our work is analogous to piecewise stationary modeling of random processes in continuous space, which has been studied in the statistics literature. However, these methods cannot be extended to graphs due to discrete nature of graphs or the fact that the definition of GWSS is not \textit{inclusive} (Section \ref{sec: challenge}). For instance, to model a non-stationary spatial process, Kim \textit{et al.} \cite{kim2005analyzing} proposed a Bayesian hierarchical model for learning optimal Voronoi tessellation while Gramacy \cite{gramacy2005bayesian} proposed an iterative binary splitting of space (treed partitioning). Piecewise stationary modeling has also been explored in time series analysis which is mainly achieved by detecting change points \cite{appel1983adaptive,last2008detecting}. Recently, there has been an effort in graph signal processing literature to detect non-stationary vertices (change points) in a random process over graph \cite{girault2017local,serrano2018graph}, however, to achieve that, authors introduce another definition for stationarity, called local stationarity, which is different from the widely-used definition of GWSS.

We evaluate our proposed model for traffic prediction on a large-scale real-world traffic dataset consisting of 4764 road segments and time-series of average travel times on each segment at a two-minute granularity over January--March 2013. We use simple piecewise linear prediction models within each cluster. The accuracy of our method is very close to the state-of-the-art deep learning method, i.e. only 0.41\% difference in the mean absolute percentage error for 10 minute predictions, rising to 0.66\% over 20 minutes. However, the learning time of our proposed method is substantially smaller, being roughly 3 hours on a personal computer while the deep learning method takes 22 hours on a server with 2 GPUs. Our contribution can be summarized as follows:
\begin{itemize}
    \item For the first time, we propose to model non-stationary random processes over graphs with piecewise stationary random processes. To that end, we propose a novel hierarchical clustering algorithm, i.e. stationary connected subgraph clustering, to extract stationary subgraphs in graph datasets.
    \item We theoretically prove that the definition of GWSS is not \textit{inclusive}. Hence, unlike many of the well-known graph clustering algorithms, adding one vertex at each training step to clusters in an stationary clustering algorithm is bound to diverge. To address this issue, we propose to use \textit{active components} as input to stationary clustering algorithm instead of vertices. 
    \item We evaluate our proposed model for traffic prediction in road networks on a large-scale real-world traffic dataset. We deploy an off-the-shelf piecewise linear prediction model for each subgraph, independently. Our method shows a very close performance to the state-of-the-art graph based deep learning methods while the computational complexity of our model is substantially smaller.
\end{itemize}

\section{Background}
\subsection{Spectral Graph Kernels}
Kernel methods, such as Gaussian process and support vector machine, are a powerful and efficient class of algorithms that have been widely used in machine learning \cite{smola2003kernels}. The core element in these methods is a positive semi-definite kernel function that measures the similarity between data points. For the data defined over vertices of a graph, spectral kernels are defined as a function of the Laplacian matrix of the graph \cite{smola2003kernels}. More specifically, given a graph $\mathcal{G}$ with Laplacian matrix $\mathbf{L}$, the kernel matrix $\mathbf{K}$ is defined as
\begin{equation}
    \mathbf{K} = \sum_{i=1}^{N} r(\lambda_i) \, \underline{u_i} \, \underline{u_i}^T.
\label{equ:spec_ker}
\end{equation}
Here, $N$ is the number of vertices of the graph, $\lambda_i$ and $\underline{u}_i$ are the $i$-th eigenvalue and eigenvector of $\mathbf{L}$, respectively, and the spectral transfer function $r : \mathbb{R}^{+} \rightarrow \mathbb{R}^{+}$ is a non-negative and decreasing function. 
Non-negativeness of $r$ assures that $\mathbf{K}$ is positive semi-definite and its decreasing property guarantees that the kernel function penalizes functions that are not smooth over graph.
Many of the common graph kernels can be derived by choosing a parametric family for $r$. For example, $r(\lambda_i) = \exp (-\lambda_i / 2\sigma^2)$ result in the well-known heat diffusion kernel.
 
Although in the majority of the previous works, spectral graph kernels are derived from the Laplacian matrix, they can also be defined as a function of the adjacency matrix of the graph \cite{avrachenkov2018similarities}. 
Adjacency based spectral kernels have the same form as \eqref{equ:spec_ker} except that the transfer function $r$ is a non-negative and ``increasing" function since eigenvectors of adjacency with \textit{large} eigenvalues have smoother transitions over graph \cite{moura2014freqanal}. We note that unlike Laplacian matrix, eigenvalues of the adjacency matrix could be negative, therefore the domain of transfer function is $\mathbb{R}$ for adjacency based kernels.

Spectral graph kernels have been widely used for semi-supervised learning \cite{zhu2005nonparametric}, link prediction \cite{kunegis2009learning} and Gaussian process regression on graphs \cite{sollich2009kernelgp,urry2013random}. In the next subsection, we show the connection between spectral graph kernels and stationary stochastic processes over graphs.

\subsection{Stationary Graph Signal Processing}\label{sec: stationary}
Graph signal processing (GSP) is a generalization of classical discrete signal processing (DSP) to signals defined over vertices of a graph (also known as \textit{graph signals}). 
Consider a graph $\mathcal{G}$ with $N$ vertices and a graph signal defined over it denoted by a $N$-dimensional vector $\underline{x}$. The graph frequency domain is determined by selecting a \textit{shift operator} $\mathbf{S}$, i.e. a $N \times N$ matrix that respects the connectivity of vertices. Applying the shift operator to a graph signal, i.e. $\mathbf{S} \, \underline{x}$, is analogous to a circular shift in classical signal processing. The Laplacian matrix $\mathbf{L}$ and the adjacency matrix $\mathbf{A}$ are two examples of the shift operator. Let us denote the eigenvalues of $\mathbf{S}$ by $\{\lambda_i\}_{i=1}^{N}$ and its eigenvectors by $\{\underline{u}_i\}_{i=1}^{N}$.  Also, assume that $\mathbf{\Lambda}=\text{diag}(\{\lambda_i\}_{i=1}^{N})$ and $\mathbf{U}$ is a matrix whose columns are $\underline{u}_i$s.
The graph Fourier transform (GFT) \cite{shuman2013emerging,moura2014bigdata} is defined as follows:

\begin{definition} [Graph Fourier transform \cite{marques2017stationary}]
Given graph $\mathcal{G}$ with shift operator $\mathbf{S}$ and a graph signal $\underline{x}$ defined over $\mathcal{G}$, the Graph Fourier Transform (GFT) of $\underline{x}$, denoted by $\hat{\underline{x}}$, is defined as 
\begin{equation}
\hat{\underline{x}} = \mathbf{U}^T \, \underline{x}.
\label{def: gft}
\end{equation}
The inverse GFT of a vector in the graph frequency domain is given by 
\begin{equation}
\underline{x} = \mathbf{U} \, \hat{\underline{x}}.
\label{def: igft}
\end{equation}
\end{definition}

It can be shown that when the Laplacian matrix is used as the shift operator, smoothness of eigenvectors over the graph is proportional to their corresponding eigenvalue \cite{shuman2013emerging}, and if the adjacency matrix is used as the shift operator, smoothness of eigenvectors over the graph is inversely proportional to their corresponding eigenvalues \cite{moura2014freqanal}. This ordering of eigenvectors provides the notion low-pass, band-pass and high-pass filters for graphs.
More specifically, filtering graph signals is defined as shaping their spectrum with a function \cite{tremblay2018design,sandryhaila2013discrete,liu2018filter}.
\begin{definition} [Graph filters \cite{marques2017stationary}]
Given a graph $\mathcal{G}$ with shift operator $\mathbf{S}$, any function $h : \mathbb{C} \rightarrow \mathbb{R}$ defines a graph filter such that
\begin{equation}
    \mathbf{H} = \mathbf{U} \, h(\mathbf{\Lambda}) \, \mathbf{U}^{T}.
\end{equation}
The filtered version of a graph signal $\underline{x}$ defined over $\mathcal{G}$ is given by $\mathbf{H} \, \underline{x}$.
\label{def: gf}
\end{definition}
\noindent An example of graph filter is the low pass diffusion filter. Assuming the Laplacian matrix is the shift operator, low pass diffusion filter $h$ is defined as $h(\lambda_i) = \exp (-\lambda_i / 2\sigma^2)$. 

An important notion in classical DSP is stationarity of stochastic processes which is the underlying assumption for many of the well known time series analysis methods. 
Stationarity of a time series means that the statistical properties of the process do not change over time. More specifically, a stochastic process is wide-sense stationary (WSS) in time if and only if its first and second moment are shift invariant in time. An equivalent definition of WSS process is that the eigenvectors of covariance matrix of a WSS process are the same as columns of discrete Fourier transform (DFT) matrix.
By generalizing the definition of WSS to stochastic processes defined on graphs, graph wide-sense stationary (GWSS) \cite{perraudin2017sspog,marques2017stationary} is defined as follows:

\begin{definition} [Graph wide-sense stationary \cite{marques2017stationary}]
A stochastic process $\underline{x}$ with covariance matrix $\mathbf{C}$ defined over vertices of a graph $\mathcal{G}$ with shift operator $\mathbf{S}$ is GWSS if and only if $\mathbf{C}$ is jointly diagonalizable with $\mathbf{S}$, i.e. $\mathbf{C}$ and $\mathbf{S}$ have the same set of eigenvectors. Equivalently, the stochastic process $\underline{x}$ is GWSS if and only if it can be produced by filtering white noise using a graph filter.
\label{def: gwss}
\end{definition}

Spectral graph kernel methods, discussed in the previous subsection, model the covariance matrix of the process over graph with the kernel matrix. The kernel matrix have the same set of eigenvectors as shift operator; hence, these methods assume that the random process they are operating on is GWSS.
It is also worth noting that GWSS reduces to WSS in case of time series. Assume that time series is a graph signal defined over a directed ring graph with adjacency matrix as the shift operator.
Multiplying the adjacency matrix with the vector of time series (applying shift operator) results in the circular shifted version of the time series by one. Therefore, definition of graph filters reduces to discrete filters in classical DSP. Moreover, the adjacency matrix of the directed ring graph is a circulant matrix whose eigenvectors are the same as the vectors that define the discrete Fourier transform. Hence, both definitions of GWSS reduce to classical definitions of WSS time series.

While GWSS is defined for random processes over graphs, in many practical scenarios, the signal on each vertex is a time series itself (time-varying graph signals). Stationarity can be further extended to time varying random processes over graphs. More specifically, a jointly wide-sense stationary (JWSS) process \cite{loukas2016stationary} is defined as follows: 

\begin{definition} [Joint wide-sense stationarity \cite{loukas2016stationary}]
A time-varying random process $\mathbf{X} = [\underline{x}^{(1)} \, \dots \, \underline{x}^{(T)}]$ defined over graph $\mathcal{G}$ with shift operator $\mathbf{S}$ is JWSS if and only if the following conditions hold:
\begin{itemize}
    \item $\mathbf{X}$ is multivariate wide-sense stationary process in time;
    \item cross covariance matrix of $\underline{x}^{(t_1)}$ and $\underline{x}^{(t_2)}$ for every pair of $t_1$ and $t_2$ is jointly diagonalizable with $\mathbf{S}$.
\end{itemize}
\end{definition}

An example of JWSS process is the joint causal model which is analogous to auto-regressive moving average (ARMA) models for multivariate WSS time-series. The joint causal model is defined as follows:
\begin{equation}
    \underline{x}^{(t)} = \sum_{i=1}^{m} \mathbf{A}_i \, \underline{x}^{(t-i)} + \sum_{j=0}^{q} \mathbf{B}_j \, \underline{\mathcal{E}}^{(t-j)},
    \label{equ: joint_cause}
\end{equation}
where $\mathbf{A}_i$ and $\mathbf{B}_j$ are graph filters and $\underline{\mathcal{E}}$ is the vector of zero mean white noise. This model can be used to predict time-varying graph signals. Given a time-varying graph signal, the parameters of the joint causal model, i.e. $\mathbf{A}_i$s and $\mathbf{B}_j$s, can be learned by minimizing the prediction error residuals which is the subject to the following nonlinear optimization problem:
\begin{equation}
    \argmin_{\{\mathbf{A}_i\}_{i=1}^{m},\, \{\mathbf{B}_j\}_{j=0}^{q}}\ ||\underline{x}^{(t)} - \widetilde{\underline{x}}^{(t)}(\{\mathbf{A}_i\}_{i=1}^{m}, \{\mathbf{B}_j\}_{j=0}^{q})||,
\end{equation}
where $\underline{x}$ is the true signal and $\widetilde{\underline{x}}$ is the output of the model. Generally, this is a computationally expensive optimization problem especially for large graphs.
Loukas \textit{et al.} \cite{loukas2017evolution} proposed an efficient way for finding the optimal filters. They proved that because graph frequency components of a JWSS process are uncorrelated, learning independent univariate prediction models along each graph frequency result in optimal filters for prediction.

While both Gaussian process regression with spectral graph kernels and joint causal model showed promising results in predicting graph signals, they suffer from two major issues.
First, in a typical real-world traffic dataset with a very large graph, the process defined over the graph is often not GWSS, which invalidates their underlying assumption. Secondly, the computational complexity of learning prediction model parameters for a large graph is formidable. To overcome aforementioned issues, we propose to split the whole graph into multiple connected disjoint subgraphs such that the process is stationary in each of them and then fitting prediction model to each individual subgraph. In the next section, first we show that this clustering is a very complex and difficult task and then we discuss our proposed computationally-efficient clustering method.

\section{Methodology}
Here, first in section \ref{sec: challenge}, we show that the definition of GWSS is not \textit{inclusive}, i.e. if a random process defined over a graph is GWSS, the subprocess defined over a subgraph is not necessarily GWSS. This shows the difficulty of clustering a large graph into stationary subgraphs.
Then, in section \ref{sec: scsc} we propose a heuristic graph clustering algorithm to find stationary subgraphs. 
In our analysis, without loss of generality, we assume that each of the time series defined on the vertices of the graph is WSS in time. Therefore we analyze random process over graphs which are not time-varying. Our analysis and proposed algorithm can simply be extended to time varying processes.
Due to lack of space we ignore some of the proofs as they can be found in the references.

\subsection{ The Challenge of Identifying Stationary Subgraphs}\label{sec: challenge}
Many graph clustering algorithms start with some initial clusters, like random clusters or every vertex as a cluster, and then they try to move vertices between clusters with the objective of maximizing a cost function. We first analyze how stationarity changes as a cluster grows or shrinks. This problem can be simplified to analyzing the stationarity of subprocesses defined over subgraphs. Assume a case where a process $\underline{x}$ with covariance matrix $\mathbf{C}$ defined over graph $\mathcal{G}$ with adjacency matrix $\mathbf{A}$ as shift operator is GWSS. The question we want to answer is that if we choose any subgraph of $\mathcal{G}$, does the subprocess defined over it is GWSS too? The ideal case is that the subprocesses defined over all of the possible subgraphs of $\mathcal{G}$ are also GWSS. We call such a process, \textit{superstationary}. More specifically, a superstatioanry random process over graph is defined as follows:

\begin{definition} [Superstationary]
A random process $\underline{x}$ defined over graph $\mathcal{G}$ is graph wide-sense superstatioanary (or superstationary for short) if and only if the following conditions hold:
\begin{itemize}
    \item $\underline{x}$ is graph wide-sense stationary over $\G$;
    \item the subprocesses defined over all of the possible subgraphs of $\mathcal{G}$ are also graph wide-sense stationary.
\end{itemize}
\end{definition}

To analyze the necessary conditions for a random process over graph to be superstationary, first we derive an equivalent definition for GWSS.
\begin{theorem}
Two square matrices are jointly diagonalizable if and only if they commute \cite{horn1990matrix}.
\end{theorem}
Combining the above theorem with definition of GWSS (Definition \ref{def: gwss}), it is straightforward to show that a random process $\x$ with covariance matrix $\C$ defined over $\G$ is GWSS if and only if $\A \, \C = \C \, \A$. 
To find a similar condition for superstationary random processes, first we review the definition of \textit{supercommuting} matrices from linear algebra literature.

\begin{definition} [Supercommuting matrices \cite{haulk1997supercommute}]
Two square matrices, $\B_1$ and $\B_2$, supercommute if and only if the following conditions hold: 
\begin{itemize}
    \item $\B_1$ and $\B_2$ commute;
    \item each of the principal submatrices of $\B_1$ and their corresponding submatrices of $\B_2$ commute.
\end{itemize}
\end{definition}

Knowing the definitions of superstationary random process and supercommuting matrices, in the following theorem, we derive the necessary and sufficient condition for a process to be superstationary.

\begin{theorem}
A random process $\x$ with covariance matrix $\C$ defined over a graph with adjacency matrix $\A$ as the shift operator is superstationary if and only if $\A$ and $\C$ supercommute.
\label{thm: ssd2}
\end{theorem}

\begin{proof}
The adjacency matrix of a subgraph is a submatrix of $\A$. Also, covariance matrix of a subprocess is a submatrix of $\C$. By definition of superstationarity, subprocesses defined over all of the possible subgraphs of $\G$ are GWSS, hence all of the corresponding submatrices of $\A$ and $\C$ commute.
\end{proof}
While the theorem above shows the necessary and sufficient condition for superstationarity, directly checking it is computationally expensive specially for large graphs. To that end, we show that given an adjacency matrix, only a class of covariance matrices are superstationary. To continue our analysis, first we review the definition of irreducible matrix and then we study the matrices that supercommute with irreducible matrices.
\begin{definition}[Irreducible matrix \cite{jeffery2000matrices}]
A matrix $\B$ is irreducible if and only if the directed graph whose weighted adjacency matrix is $\B$, is strongly connected.
\label{def: irr}
\end{definition}

\begin{theorem}
Any matrix that supercommutes with irreducible square matrix $\B_{N\times N}$ is a linear combination of $\B$ and identity matrix $\mathbf{I}_N$ \cite{haulk1997supercommute}.
\label{thm: spc_irr}
\end{theorem}

Next, we derive the sufficient and necessary condition for a process defined over a strongly connected graph to be superstationary. 
\begin{theorem}
A random process $\x$ with covariance matrix $\C$ defined over a strongly connected graph of order $N$ with adjacency matrix $\A$ as the shift operator is superstationary if and only if $\C$ is a linear combination of $\A$ and identity matrix $\mathbf{I}_N$.
\label{thm: ss_irr}
\end{theorem}
\begin{proof}
The adjacency matrix of a strongly connected graph is an irreducible matrix. Therefore, given Theorem \ref{thm: ssd2}, Definition \ref{def: irr} and Theorem \ref{thm: spc_irr}, the proof is straightforward.
\end{proof}

Assuming that the desired graph $\G$ is strongly connected\footnote{Undirected graphs are strongly connected as long as they are connected.}, which is a realistic assumption in many of the real-world networks, unless $\C$ is a linear combination of $\A$ and identity matrix, there is at least a subgraph such that the subprocess defined over it is non-stationary. This is a major challenge for identifying stationary subgraphs. Suppose that we start with a subgraph containing only one vertex and add one vertex at a time until we cover the whole graph and at each step we check whether the subprocess is stationary or not. We know that at some step the subprocess is non-stationary. But being non-stationary at a step does not necessarily mean that the subprocess is non-stationary in the next steps as we know that the process in the last step is stationary. This means that moving one vertex at a time between clusters is \textit{not} optimal and the algorithm may not converge at all. Next, we show the difference of a stationary process and a superstationary process with an example.

\begin{figure}[t]
\centering
\includegraphics[width=0.65\columnwidth, keepaspectratio]{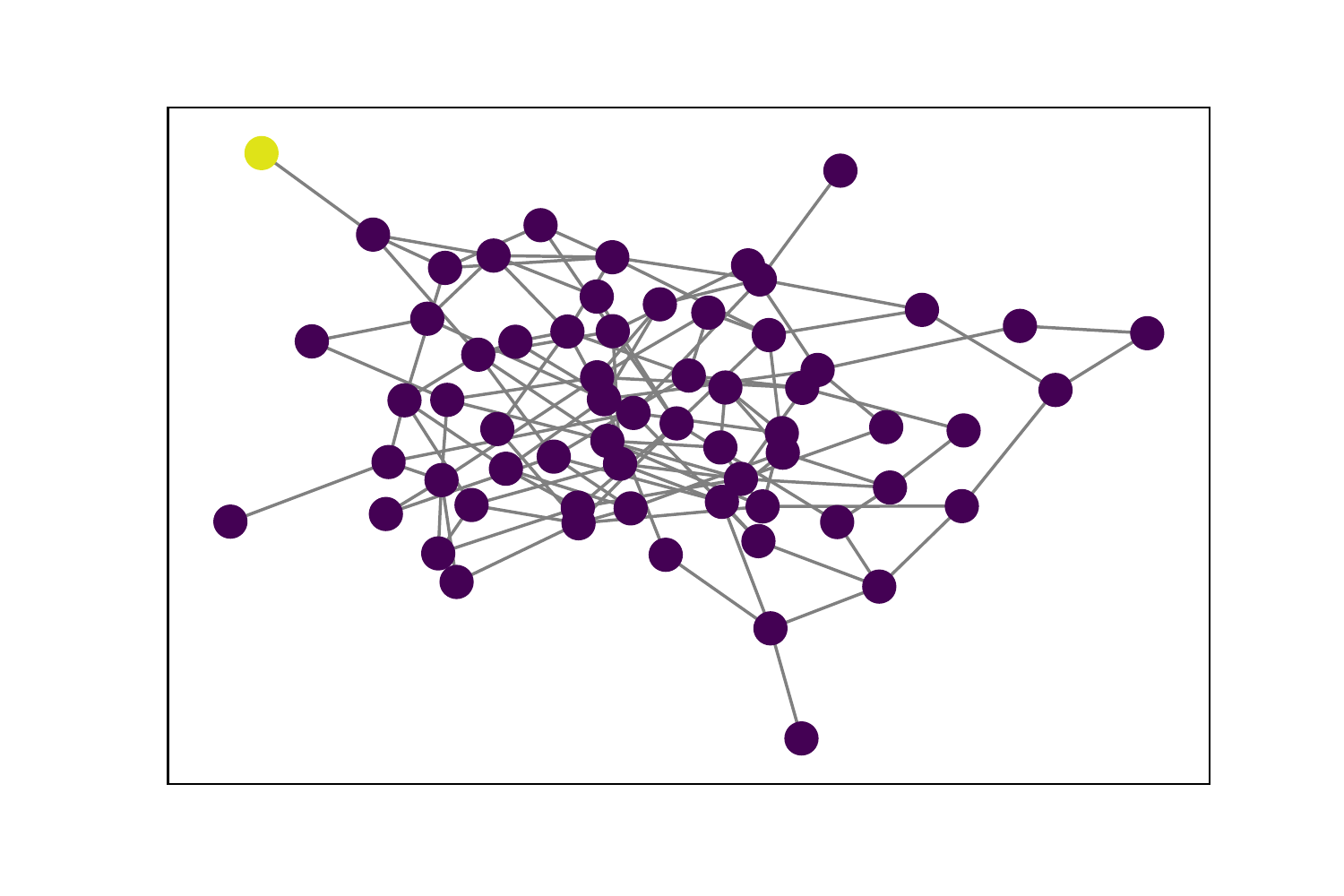}
\caption{The simulated Erd\"{o}s-Renyi graph.}
\label{fig: erdos}
\end{figure}

We generate a random Erd\"{o}s-Renyi graph with 64 nodes with edge probability of 0.06 (shown in Fig. \ref{fig: erdos}). We form two covariance matrices, a stationary and a superstationary, as follows: 
\begin{enumerate}
    \item We eigendecompose the adjacency matrix and assume that the covariance matrices of stationary and superstationary processes have the same set of eigenvectors as adjacency matrix. Let us denote the eigenvalues of the adjacency matrix by $\{\lambda_i\}_{i=1}^{N}$.
    \item We choose the eigenvalues of the stationary covariance matrix to be quadratic function (shown in Fig. \ref{fig: eigs}). More specifically, $\lambda_{i}^{s} = 2.146\times10^{-3}\ i^{2} + 1.073\times10^{-5}$ for $i \in \{1, \dots, 64\}$.
    \item We choose the the eigenvalues of the superstationary covariance matrix to be linear function of eigenvalues of the adjacency matrix (shown in Fig. \ref{fig: eigs}). More specifically, $\lambda_{i}^{ss} = 0.5\ \lambda_{i} + 2$ for $i \in \{1, \dots, 64\}$.
\end{enumerate}
\vspace{3pt}

We start with a subgraph consisting of a randomly chosen node (the yellow node in Fig. \ref{fig: erdos}) and its immediate neighbors and compute the stationarity ratio of the processes over this subgraph. We keep increasing the size of the subgraph by adding one-hop neighbors of the nodes in the subgraph (at the current step) to it and compute the stationarity ratio at each step. The results are shown in Fig. \ref{fig: stratios}. As we expected, the superstationary process is completely stationary on all of subgraphs while the stationary ratio of the stationary process could decrease to less than 0.7 for some of the subgraphs.
This, indeed, shows the challenge of stationary graph clustering and proves that moving one vertex at a time between clusters could cause the algorithm to diverge from the optimal solution.
In the next subsection we propose a heuristic approach to overcome this problem.

\begin{figure}[t]
\includegraphics[clip,
trim=0cm 0cm 1cm 1cm,
width=\columnwidth, 
height=0.25\textheight,
keepaspectratio]{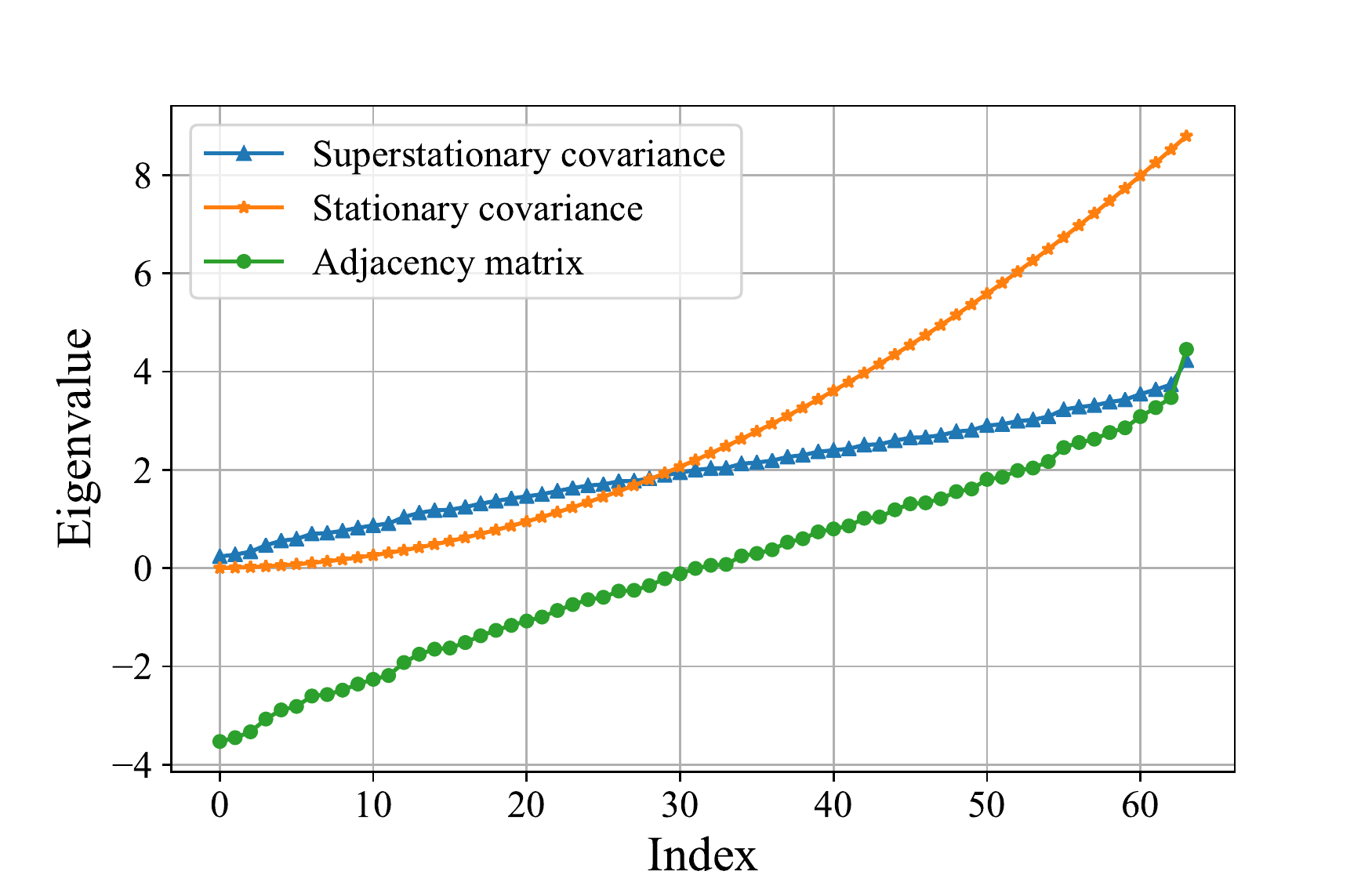}
\caption{Eigenvalues of the simulated processes and the adjacecny matrix of the graph.}
\label{fig: eigs}
\end{figure}

\begin{figure}[b]
\includegraphics[clip,
trim=0cm 0cm 1cm 1cm,
width=\columnwidth, 
height=0.25\textheight, keepaspectratio]{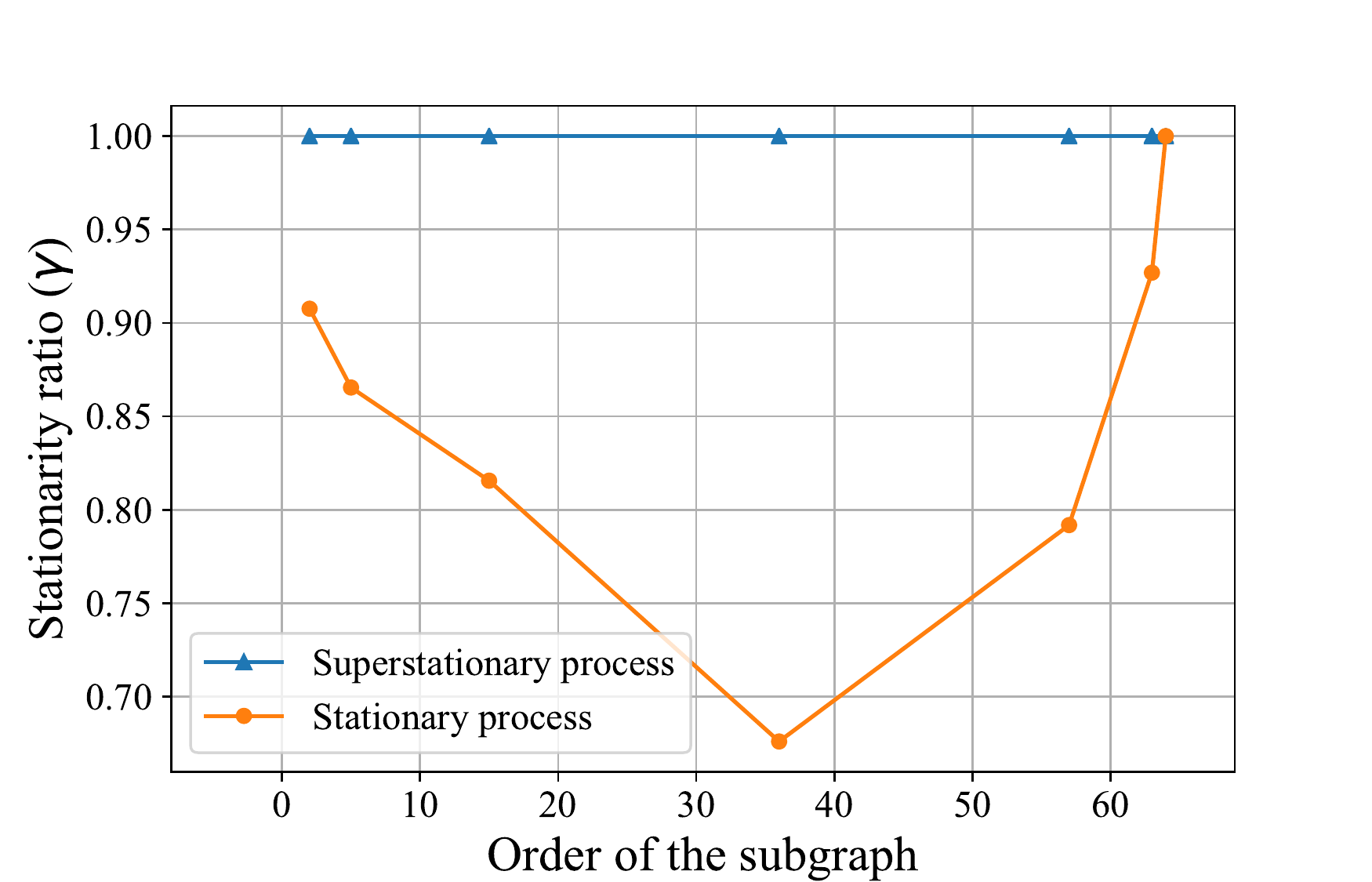}
\caption{Stationary ratio of the simulated processes over subgraphs.}
\label{fig: stratios}
\end{figure}

\begin{figure*}[!t]\label{fig: congpat}
\centering
\includegraphics[width=\textwidth,
keepaspectratio]{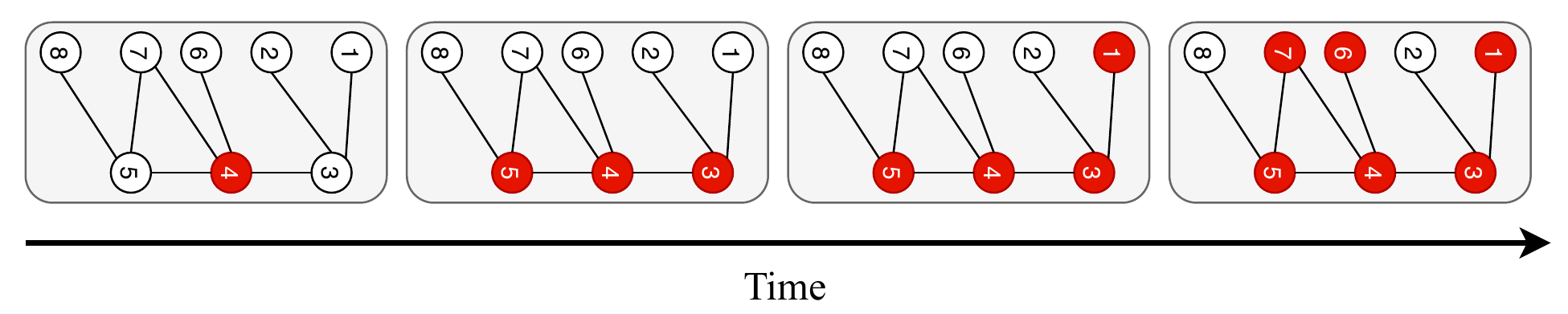}
\caption{Visualization of an active component of a graph. At the fist snapshot, node 4 becomes active and as time goes by the activity spreads in the network. The red nodes at each snapshot represents active nodes. The subgraph containing nodes 1, 3, 4, 5, 6 and 7 is an active component of the network.}
\label{fig:congpat}
\end{figure*}

\subsection{Stationary Connected Subgraph Clustering}\label{sec: scsc}
We are interested in modeling a non-stationary random process over a graph using a piecewise stationary process over the same graph. 
Therefore, the \textit{stationary graph clustering} problem is defined as follows:
\begin{problem}
Given a graph $\G$ and a non-stationary stochastic process $\x$ defined over $\G$, partition the graph into $k$ disjoint connected subgraphs $\{\G_1,\G_2,\dots,\G_k\}$ such that each of the subprocesses $\{\x_1,\x_2,\dots,\x_k\}$ defined over subgraphs are GWSS.
\end{problem}

We propose a heuristic algorithm, stationary connected subgraph clustering (SCSC), to solve the problem above. To the best of our knowledge, we address this problem for the first time.

As we discussed in the previous subsection, adding one vertex at a time to identify stationary subgraphs is problematic. Therefore, we propose a clustering algorithm whose inputs are \textit{active components} rather than vertices. Active components (ACs) are spatial localization of activity patterns made by the process over the graph (see Fig. \ref{fig:congpat}). Spatial spreading of congestion in transportation networks and spreading pattern of rumor in social networks are two examples of active components.
First, we clarify the mathematical definition of a vertex being \textit{active}. If magnitude of the signal defined over a vertex exceeds a threshold, i.e. $\x[i] \geq \alpha$, we say that the vertex is active. For example, if the travel time of a road exceeds a threshold it means that it is congested or active.
Knowing the mathematical definition of active vertex, active component is defined as follows:
\begin{definition} [Active component]
We say that the vertex $i$ and vertex $j$ are in the same active component if and only if at some time $t$ both of the following conditions hold: 
\begin{itemize}
    \item $(i,j) \  \text{or} \ (j,i)$ is an edge of the graph;
    \item vertex $i$ is active at time $t$ and vertex $j$ is active at time $t$ or $t-1$ or vice versa.
\end{itemize}
\end{definition}
ACs can be viewed as subgraphs on which a diffusion process takes place. We already know that diffusion processes are stationary processes. Therefore, we expect that the subprocesses defined over ACs to be stationary.
To obtain ACs from historical observations, we propose the active components extraction iterative algorithm.

\begin{algorithm}[t]
  \caption{Active Components Extraction}
  \begin{algorithmic}[1]
      \State $\textbf{Input: } \G,\ \alpha,\ \Y\in \mathbb{R}^{N\times T}$
      \State \textbf{Notations:}
      \State \ \ \ \ $\G = \text{network graph}$
      \State $\ \ \ \ \alpha = \text{threshold on signal determining active vertices}$
      \State $\ \ \ \ \Y = \text{historical observation matrix}$
      \State $\textbf{Initialize:}$  
      \State $\ \ \ \ \text{AC} \gets \{\}, \text{AC}_t \gets \{\}, \text{AC}_{t-1} \gets \{\}$
      \For{$i = 1:T$}
        \State $\mathcal{R} \gets \text{argwhere}(Y(:,i)<\alpha)$
        \State $\G_{sub} \gets \text{remove }\mathcal{R}\text{ and connected edges from }\G$
        \State $\text{AC}_t \gets \text{connected components of } \G_{sub}$
        \For{$j = 1:\textit{len}(\text{AC}_{t-1})$}
            \State $\text{indicator} \gets \text{False}$
            \For{$k = 1:\textit{len}(\text{AC}_{t})$}
                \State $\text{AC}_t(k) \gets \text{one hop expansion of AC}_t(k)$
                \If{$\text{(AC}_t(k)\ \cap\ \text{AC}_{t-1}(j)) \neq \varnothing$}
                    \State $\text{AC}_t(k) \gets \text{(AC}_{t}(k)\ \cup\ \text{AC}_{t-1}(j))$
                    \State $\text{indicator} \gets \text{True}$
                    \EndIf
                \EndFor
                \If{$\text{indicator == False}$}
                    \State $\text{append AC}_{t-1}(j)\text{ to AC}$
                    \EndIf
            \EndFor
      \EndFor
      \State \textbf{return} AC 
  \end{algorithmic}
  \label{alg: csp}
\end{algorithm}

\begin{algorithm}[t]
  \caption{Stationary Connected Subgraph Clustering}\label{scp}
  \begin{algorithmic}[1]
      \State $\textbf{Input: } \G,\ \gamma_{th},\ \C,\ \text{AC},\ \theta,\ \D$
      \State \textbf{Notations:}
      \State \ \ \ \ $\G = \text{network graph}$
      \State $\ \ \ \ \gamma_{th} = \text{threshold on stationarity ratio}$
      \State $\ \ \ \ \C = \text{covariance matrix}$
      \State $\ \ \ \ \text{AC} = \text{set of active components}$
      \State $\ \ \ \ \theta = \text{number of clusters}$
      \State $\ \ \ \ \D = \text{matrix of distances between active components}$ 
      \State $\textbf{Initialize:}$  
      \State $\ \ \ \ \text{CLS} \gets \text{AC},\ \text{NL} \gets \{\},\ n \gets |\text{AC}|$
      \While{$n > \theta$}
        \State $d_{min} \gets \text{min}\ (\D)$
        \State $[i,\ j] \gets \argmin\ (\D)$
        \If{$d_{min} < 2 \textbf{ and } [i,\ j] \nsubseteq \text{NL}$}
            \State $\text{CT} \gets (\text{AC}(i)\ \cup\ \text{AC}(j))$
            \State $\G_{sub} \gets \text{induced subgraph of }\G\text{ for nodes in CT}$
            \State $\C_{sub} \gets \C[\text{CT},\ \text{CT}]$
            \State $\gamma \gets \text{compute stationarity ratio using }\C_{sub}\ \& \ \G_{sub}$
            \If{$\gamma \geq \gamma_{th}$}
                \State remove AC$(i)$ and AC$(j)$ from CLS
                \State insert CT to CLS
                \State update $\D$ in single linkage manner
                \State $n \gets n-1$
                \Else \ append $[i,\ j]$ to NL
            \EndIf
        \EndIf
        \If{$\lvert\text{NL}\rvert == \lvert\text{CLS}\rvert$}
            \State break
        \EndIf
      \EndWhile
      \State \textbf{return} CLS 
  \end{algorithmic}
  \label{alg: scsc}
\end{algorithm}

The overview of Algorithm \ref{alg: csp} is that at each time step, first we create the \textit{active} graph by removing the inactive nodes and all the edges connected to these nodes from the whole network graph. Then we go through time and if activity patterns at time $t-1$ are propagating to time $t$, the active components are merged into one. In fact, the algorithm finds weakly connected components of the strong product of spatial graph and time-series graph when spatio-temporal nodes with signal amplitude less than $\alpha$ and all of connected edges to them are removed.

Before we describe the second part of the algorithm in which the adjacent ACs are merged and form clusters, 
we review the measure of stationarity defined in \cite{loukas2016stationary}. First, we form the following matrix 
\begin{equation}
\P = \U^T \, \C \, \U,
\end{equation}
where $\U$ is the matrix of eigenvectors of the shift operator and $\C$ is the covariance matrix of the random process.
Stationarity ratio, $\gamma$, of a random process over graph is defined as follows:
\begin{equation}
    \gamma = \frac{||\text{diag}(\P)||_2}{||\P||_F},
\end{equation}
where $||.||_F$ represents the Frobenius norm and diag(.) denotes the vector of diagonal elements of the matrix. In fact, $\gamma$ is a measure of diagonality of matrix $\P$. If a process is GWSS then eigenvectors of covariance matrix and shift operator of the graph are the same hence $\P$ is diagonal and $\gamma$ equals to one. Diagonal elements of $\P$ form the power spectral density of the GWSS process. 

Another definition that we need to move forward to the next step is the distance between two ACs. The distance between two ACs is the minimum of the shortest path distances between all pairs of nodes $(v_1,v_2)$ where $v_1 \in  \text{AC}_1$ and $v_2 \in  \text{AC}_2$. 
We note that two ACs could have a common node, hence the distance between two ACs could be zero.
Knowing all necessary definitions, the psudocode for SCSC algorithm is described in Algorithm 2.

The intuition behind SCSC is that it is most likely that adjacent ACs belong to the same diffusion process. SCSC merges adjacent ACs if after merger the stationarity ratio is larger than some threshold.
This definition makes sense since it is repeatedly observed that activity on one vertex easily causes or serves as a result of activity in the adjacent vertex with some time difference.
Therefore, we expect that the process defined over the output clusters of SCSC to be stationary. 
SCSC is a hierarchical clustering algorithm with some extra conditions. The overload caused by these conditions are small compared the complexity of hierarchical clustering. The conditions only needs eigendecomposition of small matrices (usually less than $50\times50$) which is negligible. 

While SCSC algorithm, as defined in Algorithm \ref{alg: scsc}, extracts subgraphs that are GWSS, with a simple modification to the algorithm, we can identify subgraphs that are JWSS. Assuming that a time varying random process $\X$ is WSS in time, the necessary condition for $\X$ to be JWSS is that all of the lagged auto-covariance matrices are jointly diagonalizable with $\S$. Hence, lines 17 to 19 in Algortihm \ref{alg: scsc} (the merger condition of clusters in SCSC) are replaced by the following lines:
\begin{empheq}[left=\empheqlvert\;\;]{align*}
    &\textbf{for}\, \, \, l = 0:q \, \, \, \textbf{do}\\
    &\, \, \, \, \, \, \, \, \C^{(l)}_{sub} \gets \C^{(l)}[\text{CT},\ \text{CT}] \\
    &\, \, \, \, \, \, \, \, \gamma_l \gets \text{compute stationarity ratio using }\C^{(l)}_{sub}\ \& \ \G_{sub}\\
    &\textbf{if}\, \, \, \gamma_l \geq \gamma_{th} \, \, \, \text{for}\, \, \, l = 0:q \, \, \, \textbf{then}
\end{empheq}
where $\C^{(l)}$ is the cross covariane between $\x^{(t)}$ and $\x^{(t-l)}$ and $q$ is a hyper-parameter denoting maximum lag. Note that both $\C^{(l)}$ and $q$ are inputs of the algorithm.

\subsection{Application to Traffic Prediction}
To examine our model, we use our proposed approach for travel time prediction in road networks. First, we use line graph of transportation network to map roads into vertices of a directed graph. Then, we use travel time index (TTI) \cite{pu2011analytic} to extract active components from historical data.
TTI of a road segment is defined as its current travel time divided by the free flow travel time of the road segment or, equivalently, the free flow speed divided by the current average speed. TTI can be interpreted as a measure of severity of congestion in a road segment.

We also propose using the joint causal model \eqref{equ: joint_cause} for each cluster separately.
To capture the dynamic and non-linear behavior of traffic, we propose using a piecewise linear model as prediction model for time series along each graph frequency. More specifically, threshold auto-regressive (TAR) \cite{chan1990testing,chan1986estimating} models are piece-wise linear extension of linear AR models. The TAR model assumes that a time-series can have several regimes and its behavior is different in each regime. The regime change could be triggered either by past values of the time-series itself (self exciting TAR \cite{petruccelli1986portmanteau}) or some exogenous variable. This model is a good candidate to model the traffic behavior; once a congestion happens, the dynamics of time-series changes. A TAR model with $l$ regimes is defined as follows:
\begin{equation*}
y_{t}=
\begin{cases}
	\sum_{i=1}^{m_1} a_{i}^{(1)}y_{(t-i)}+b^{(1)}\epsilon_t, & \beta_0 < z_t < \beta_1\\
	\sum_{i=1}^{m_2} a_{i}^{(2)}y_{(t-i)}+b^{(2)}\epsilon_t, & \beta_1 < z_t < \beta_2\\
	&\vdots\\
	\sum_{i=1}^{m_l} a_{i}^{(l)}y_{(t-i)}+b^{(l)}\epsilon_t, & \beta_{l-1} < z_t < \beta_{l}\\
\end{cases}
\end{equation*}
where $z$ is the exogenous variable. A natural choice of exogenous variable for traffic is TTI, because it shows severity of congestion.

\section{Numerical Results And Discussion}
\subsection{Dataset}
The traffic data used in this study originated from the Dallas-Forth Worth area, with a graph comprising 4764 road segments in the highway network. The data represented time-series of average travel times as well as average speed on each segment at a two-minute granularity over January--March 2013. 
The data was used under licence from commercial data provider, which produces this data applying proprietary methods to a number of primary sources including location data supplied by mobile device applications. Fig. \ref{fig: map} shows the map of the road network. Missing values form dataset were imputed by moving average filters. In all experiments, we used 70\% of data for training, 10\% for validation and 20\% for testing.

\begin{figure}[!t]
\centering
\includegraphics[width=\columnwidth, height=0.5\textheight, keepaspectratio]{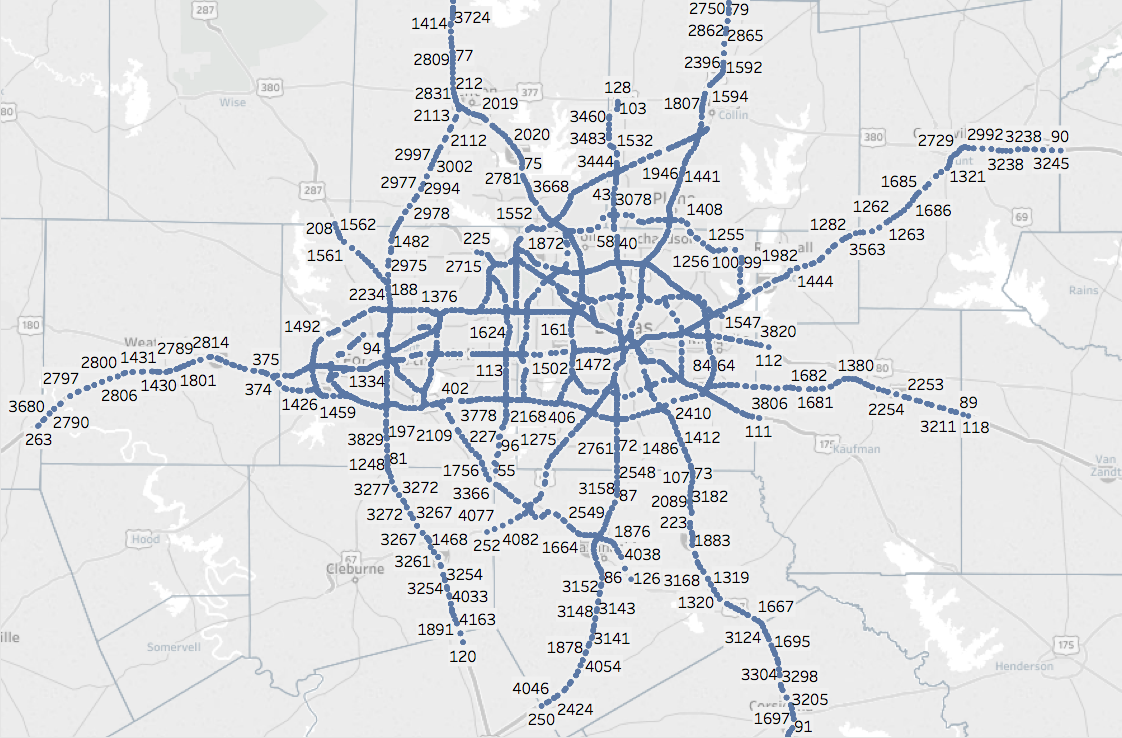}
\caption{The map of the road network in Dallas.}
\label{fig: map}
\end{figure}

\subsection{Experimental Setup}
\subsubsection{Our Proposed Method}
After removing daily seasonality from time series, difference transformation was used to make the time series along each vertex WSS. Augmented Dicky-Fuller test was used to check for stationarity in time.
We used 1.7 as threshold ($\alpha = 1.7$) for TTI to detect congestion in a road. 117621 active components were extracted from the training data using Algorithm \ref{alg: csp}. We eliminated ACs with less than 5 vertices which reduced the number of ACs to 52494. 

In our clustering algorithm, we used combinatorial Laplacian matrix of directed graphs \cite{chung2005laplacians}, which is defined as follows
\begin{equation*}
    \L = \frac{1}{2}(\D_{\text{out}} + \D_{\text{in}} - \A - \A^{T}),
\end{equation*}
as the shift operator. In the equation above, $\D_{\text{out}}$ and $\D_{\text{in}}$ represent the out-degree and in-degree matrices, respectively. The in-degree (out-degree) matrix is a diagonal matrix whose $i$-th diagonal element is equal to the sum of the weights of all the edges entering (leaving) vertex $i$. 

We initiated the SCSC algorithm (Algorithm \ref{alg: scsc}) by setting minimum stationarity ratio to 0.9 and number of clusters to 150. The number of vertices in clusters are between 24 to 61. Prediction models were learned for each cluster independently in the graph frequency domain. TAR models consist of three different regimes.
The exogenous threshold variable for our proposed prediction models, joint causal model with TAR (JCM-TAR), is the sum of TTI values of all nodes in a cluster at each time step. This value shows how congested the whole cluster is. We implemented our clustering algorithm in Python and used tsDyn package \cite{dinarzo2019tsdyn} in R 
to learn the parameters of (threshold) auto-regressive models.

\begin{figure}[t]
\centering
\includegraphics[clip,
height=0.27\textheight, keepaspectratio]{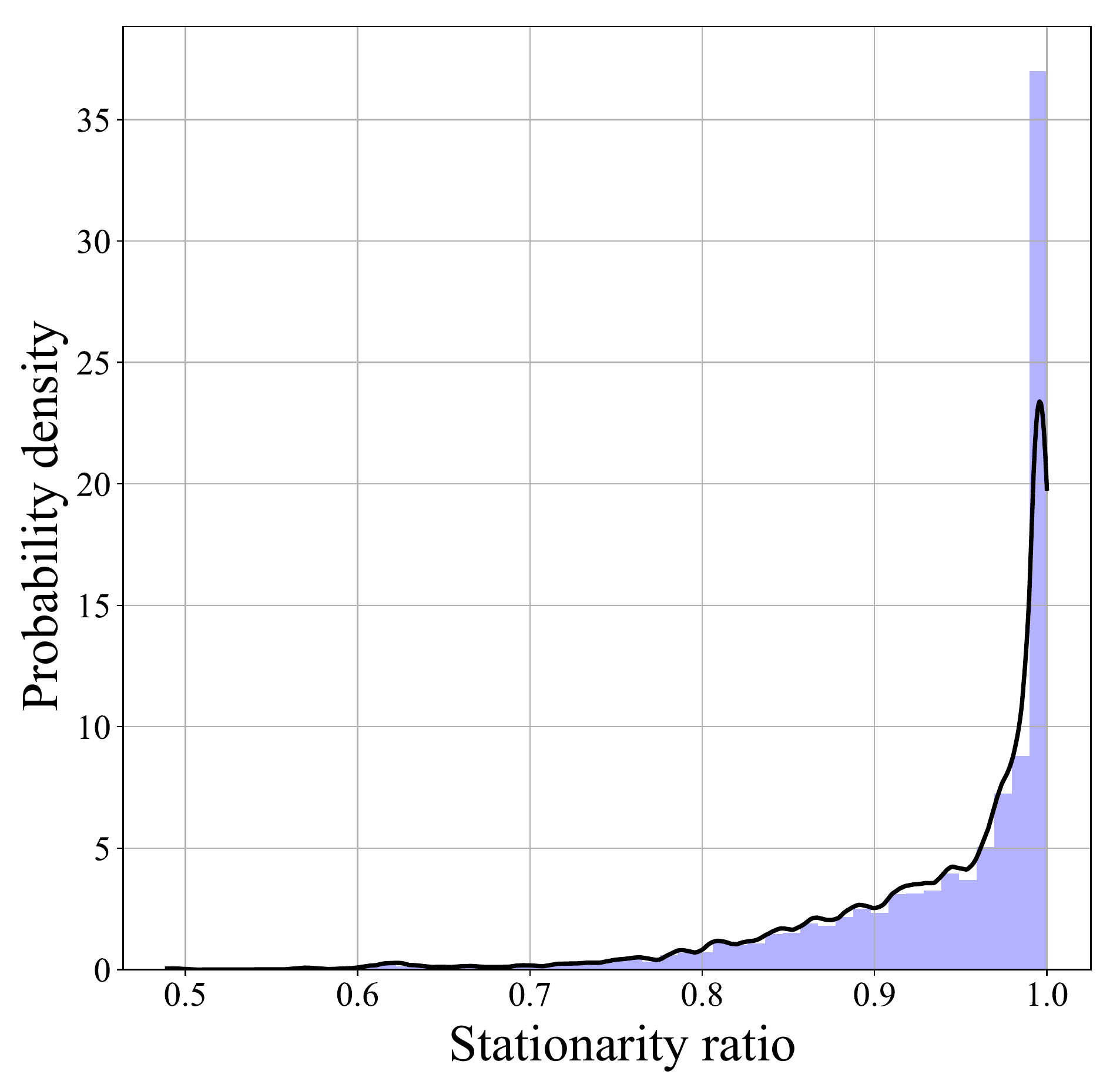}
\caption{The histogram of stationarity ratio of subprocesses defined over active components.}
\label{fig: pd}
\end{figure}

\begin{table*}[t]
\centering
\label{tab:la_comparison}
\caption{Traffic prediction performance of our proposed method and baselines.}
\resizebox{1.3\columnwidth}{!}{
\begin{tabular}{c|c|ccccc}

\toprule
$T$ & Metric & ARIMA  & DCRNN & STGCN & JCM-AR  & JCM-TAR\\ \hline
\midrule
\multirow{3}{*}{10 min}& MAE  & 2.7601 & 1.1764 & 1.0437 & 2.0520 & 1.2732 \\ 
& RMSE & 5.3204 & 2.7553  & 2.6370 & 4.8620  & 2.9993 \\ 
& MAPE & 5.10\%  & 2.68\% & 2.44\% & 4.28\% & 2.85\% \\ 
\hline 
\multirow{3}{*}{14 min}& MAE  & 3.3427 & 1.3837 & 1.1586 & 3.1187   & 1.7415 \\ 
& RMSE & 6.5432     & 3.1186   & 2.8682 & 5.3481   & 3.5231 \\ 
& MAPE & 8.32\%  & 3.22\% & 3.11\% & 7.84\% & 3.51\% \\ 
\hline
\multirow{3}{*}{20 min}& MAE  & 4.9971    &  1.5948  & 1.3911 & 4.3873 &  1.9926 \\ 
& RMSE & 10.8734 &  3.5287  & 3.3406 & 10.0381  &  4.0154 \\ 
& MAPE & 13.65\%  &  3.83\% & 3.67\% & 12.75\% & 4.33\% \\ 
\bottomrule

\end{tabular}
}
\label{tbl: res}
\end{table*}

\subsubsection{Baselines}
To establish accuracy benchmarks, we consider four other prediction models. The first scheme is to build independent auto-regressive integrated moving average (ARIMA) \cite{mccleary1980applied} models for time series associated with each road. After removing daily seasonality form data, ARIMA models with five moving average lags and one auto-regressive lag were learned. This naive scheme ignores the spatial correlation of adjacent roads. The second benchmark scheme is a joint causal model with non-adaptive AR models (JCM-AR) in graph frequency domain of each cluster. We used the same clusters as our proposed method discussed in previous subsection. 

The third scheme is the diffusion convolutional recurrent neural network (DCRNN) \cite{li2018diffusion}. This deep learning method uses long short-term memory (LSTM) cells on top of diffusion convolutional neural network for traffic prediction. The architecture of DCRNN includes 2 hidden RNN layers, each layer includes 32 hidden units. Spatial dependencies are captured by dual random walk filters with 2-hop localization. The last baseline is the spatio-temporal graph convolutional neural network (STGCN) \cite{yu2018stpc}. This method uses graph convolutional layers to extract spatial features and deploys one dimensional convolutional layers to model temporal behavior of the traffic. We used two graph convolutional layers with 32 Chebyshev filters. We also used two temporal gated convolution layers with 64 filters each. We used the Python implementations of DCRNN and STGCN provided by the authors.

\subsection{Discussion}
To check our hypothesis that active components are stationary, we look at the probability density of stationarity ratios of data defined over of active components (showed in Fig. \ref{fig: pd}). We note that most of the ACs have stationarity ratio of more than 0.8 which approves our hypothesis. 

We also compare our proposed SCSC with spectral clustering \cite{ng2002spectral} and normalized cut clustering \cite{dhillon2004kernel} algorithms. The average stationarity ratio of clusters for spectral clustering and normalized cut clustering are 0.4719 and 0.3846, respectively, while our SCSC algorithm produces clusters with stationarity ratio of more than 0.9. SCSC shows superior performance in finding stationary clusters as other clustering algorithms have different objectives.

Table \ref{tbl: res} shows the prediction accuracy of the proposed method and baselines for 10 minutes, 14 minutes and 20 minutes ahead traffic forecasting.
The metrics used to evaluate the accuracy are mean absolute error (MAE), mean absolute percentage error (MAPE) and root mean squared error (RMSE); see \cite{li2018diffusion}.

ARIMA shows the worst accuracy because it does not capture spatial correlation between neighboring roads and it cannot capture dynamic temporal behavior of traffic using a linear model. JCM-AR improves the the performance of ARIMA because it uses spatial correlation implicitly. The difference between these two models shows the importance of spatial dependencies. The improvement is more substantial for long term predictions.
JCM-TAR model with three regimes shows the second best performance. This shows that temporal dynamic behavior of traffic is changing once a congestion happens in the network because a congestion can change the statistics of the process drastically.
We note that the performance of JCM-TAR is very close to DCRNN especially for short term traffic prediction. However, the small enhancement of DCRNN comes with a huge increase in model complexity and scalability.

To learn the parameters of DCRNN and STGCN, we used a server with 2 GeForce GTX 1080 Ti GPUs which took more than 22 and 11 hours to converge, respectively. Doubling the number of hidden units in the RNN layers of DCRNN to 64 units, the server ran out of memory. However, it took less than 3 hours for our method to converge using a computer with 2.8GHz Intel Corei7 CPU and 16GB of RAM to simulate our method. The active components extraction took less than an hour and the SCSC algorithm converged in less than 2 hours. Estimation of our proposed prediction models are very fast because we used univariate piecewise linear models both in time and space which can be implemented completely in parallel.

Another advantage of our method is the complexity of parameter tuning.
While the parameters in our model are interpretive and most of them do not need any tuning, deep method could be very sensitive to network architecture. Extensive search over different parameters is the common way in deep learning to find the best architecture. It took us days to tune the parameters of DCRNN.
Deep learning models can predict the traffic with very high accuracy \cite{lv2015deep,yu2018stpc,yu2017deep,wu2016deep,ma2017deep,li2018diffusion} but most of them ignore complexity and scalability of the model to real world big datasets. 
However, a carefully designed prediction model, for short term traffic prediction, can perform as well as deep learning with less complexity and better scalability.

\section{Conclusion}
In this work, we introduced a new method for modeling non-stationary random processes over graphs. We proposed a novel graph clustering algorithm, called SCSC, in which a large graph is partitioned into smaller disjoint clusters such that the process defined over each cluster is assumed to be stationary and independent from other clusters. Independent prediction models for each cluster can be deployed to predict the process. Numerical results showed that combining our piecewise stationary model with a simple piecewise linear prediction model shows comparable accuracy to graph based deep learning methods for traffic prediction task. More specifically, the accuracy of our method is only 0.41 lower in mean absolute percentage error to the state-of-the-art graph based deep learning method, while the computational complexity of our model is substantially smaller, with computation times of 3 hours on a commodity laptop, compared with 22 hours on a 2-GPU array for deep learning methods.

\section{Acknowledgements}
This material is based upon work supported by the National Science Foundation under Grant ENG-1839816.

\bibliographystyle{IEEEtran}
\bibliography{ijcai19}

\end{document}